\documentclass[conference]{IEEEtran}
\IEEEoverridecommandlockouts

\usepackage{amsmath}
\usepackage{bm}
\usepackage{mdframed}
\usepackage{amsthm}

\usepackage[utf8]{inputenc}
\usepackage{graphicx}
\usepackage[colorinlistoftodos]{todonotes}
\usepackage{amssymb}
\usepackage{longtable}
\usepackage{nohyperref}
\hypersetup{bookmarks=false}
\usepackage[noadjust]{cite}
\allowdisplaybreaks
\usepackage{lineno}
\usepackage{lipsum}
\usepackage{multicol}
\usepackage{subcaption}
\usepackage{enumitem}
\newtheorem{theorem}{Theorem}

\usepackage{xcolor}
\usepackage[vlined,linesnumbered,ruled,resetcount]{algorithm2e}

\SetKwInput{KwInput}{Input}                
\SetKwInput{KwOutput}{Output}              
\SetKwInput{KwInitialize}{Initialize}                

\renewcommand{\baselinestretch}{0.98} 

\begin{document}
 \title{On the Optimization of Multi-Cloud Virtualized Radio Access Networks}

\author{\IEEEauthorblockN{Fahri Wisnu Murti\IEEEauthorrefmark{1}, Andres Garcia-Saavedra\IEEEauthorrefmark{2}, Xavier Costa-Perez\IEEEauthorrefmark{2}, George Iosifidis\IEEEauthorrefmark{1}}\\
	\vspace{-4.5mm}
	\IEEEauthorblockA{
		\IEEEauthorrefmark{1}School of Computer Science and Statistics, Trinity College Dublin\\
		\IEEEauthorrefmark{2}NEC Laboratories Europe, Heidelberg, Germany}
		\thanks{{This work was supported by Science Foundation Ireland (SFI) under Grant No. 17/CDA/4760. This work has been partially supported by EC H2020 5GPPP 5Growth project (Grant 856709). This publication has emanated from research supported in part by a research grant from SFI and is co-funded under the European Regional Development Fund under Grant Number 13/RC/2077. 
		***This preprint is to appear in \textit{Proc. of IEEE International Conference on Communications (ICC) 2020.}}}%
}

\maketitle
\thispagestyle{plain}
\pagestyle{plain}

\begin{abstract}

We study the important and challenging problem of virtualized radio access network (vRAN) design in its most general form. We develop an optimization framework that decides the number and deployment locations of central/cloud units (CUs); which distributed units (DUs) each of them will serve; the functional split that each BS will implement; and the network paths for routing the traffic to CUs and the network core. Our design criterion is to minimize the operator's expenditures while serving the expected traffic. To this end, we combine a linearization technique with a cutting-planes method in order to expedite the \emph{exact solution} of the formulated problem. We evaluate our framework using real operational networks and system measurements, and follow an exhaustive parameter-sensitivity analysis. We find that the benefits when departing from single-CU deployments can be as high as 30\% for our networks, but these gains diminish with the further addition of CUs. Our work sheds light on the vRAN design from a new angle, highlights the importance of deploying multiple CUs, and offers a rigorous framework for optimizing the costs of Multi-CUs vRAN. 

\end{abstract}
\IEEEpeerreviewmaketitle



\section{Introduction}

\subsection{Motivation}

The Cloud Radio Access Network (C-RAN) has recently emerged as a promising solution for building low-cost high-performance RAN in 5G+ systems. Pooling the base station (BS) functions in central servers within the RAN, reduces their implementation cost and offers unprecedented performance gains. For instance, C-RAN can support Coordinated Multi-Point (CoMP) transmissions, flexible spectrum management, and intelligent interference control \cite{5g_white,ngmn_white}. The last few years both academia and industry have been working to standardize C-RANs and expedite their adoption \cite{checko_cloudran}. However, there are still several issues that need to be addressed before we can fully reap the benefits of these systems. 

One of the key challenges is the design of optimal C-RAN architectures. This remains an open and difficult problem for many reasons. First, it is clear today that fully centralized RANs are not always implementable as they require expensive high-capacity fronthaul networks. Hence, the latest proposals enable operators to determine the centralization level (or, \textit{functional split}), by selecting which BS functions will be hosted at the central/cloud units (CUs) and which will be kept at the distributed units (DUs). The term virtualized RAN (vRAN) has been coined to describe these architectures, which in their most flexible version allow even a different split for each BS \cite{nokia_anyhaul}. However, selecting the level of centralization is an intricate problem, as each functional split creates different computation load for the CUs and DUs, and different data transfer loads for the network.

Furthermore, some vRANs might need to deploy multiple CUs in order to increase the number of centralized BS functions. It is thus necessary to decide the number and deployment location of the CUs, and which DUs each of them should serve. Nevertheless, these \emph{assignment} decisions are affected by the functional split that is selected for each BS. On top of that, while initially the vRAN fronthaul was designed using point-to-point connections (CPRI lines), there is nowadays consensus that these links should be replaced by packet-switched networks where links can be shared by multiple flows. This is more cost-effective, yet it compounds the vRAN design problem as one has, additionally, to decide the \emph{routing} for each BS and each split.

\textbf{Our goal in this work is:} to tackle the vRAN design problem in its most general form by optimizing jointly the number and location of CUs; assignment of DUs to CUs; functional split for each BS; and routing of data in the network.

\subsection{Related Work}

The idea of C-RAN was followed, early-on, by the suggestion for implementing the BS functions in common hardware (cloudification) \cite{yin-c-ran}; and soon after \cite{function_split_bells} analyzed the pros and cons of different splits. A detailed study of the vRAN split specifications can be found in \cite{smallcell}, while \cite{cost_vm, pooling-conext} studied the cost-efficiency gains of splits, and \cite{impact_packet} the effect of packetization. However, there are few works that optimize the split configuration, see \cite{function_split_survey}. The authors in \cite{davit_flex5g} select splits to minimize inter-cell interference and fronthaul load, \cite{adaptive_alba, yahya_adaptive} consider adaptive split configuration, \cite{slicing_ojaghi} studied jointly the splitting and RAN slicing; and our previous work \cite{fluidran_andrea} proposed joint routing and split selection. 

The above works do not consider the availability of multiple CUs, even less so the need to determine their deployment location. However, this is a key step in vRAN design. In \cite{scalable_mharsi}, the authors consider a tree-like structure with pre-deployed CUs and explore, heuristically, minimum-cost splits; while \cite{osama_provisioning} selects also the location of CUs and formulates (but does not solve) a min-cost design problem with fixed (non-optimized) splits. Finally, in our recent work \cite{wizhaul_andrea} we assumed multiple and fixed number of CUs but with predetermined DU-CU assignments. Unlike these interesting works, here we decide the deployment of CUs, determining also their number, and select which DUs each of them will serve and which splits will be realized. This complicates substantially the problem, but tackles a very important and practical design dimension in vRANs. Modeling-wise, the problem is reminiscent to server placement problems \cite{leonidas-TPDS} and VNF chain embeddings \cite{raz-vnf-infocom15}. However, here the placement of functions determines the data volume and transfer delay bounds, and hence which routing options are feasible.


\subsection{Contributions}

We introduce an analytical framework by deciding: the number of CUs, their deployment locations, the DUs that each of them should serve, the split for each BS, and the routing path for each data flow. We formulate a mathematical problem that optimizes jointly these decisions by using a measurement-based 3GPP-compliant system model. The objective is to minimize the vRAN operational expenditures, including function implementation and data routing costs. Our framework is general enough and can be tailored to different system architectures or include other objective criteria such as the desirable centralization level (see \cite{wizhaul_andrea} for definition).

The resulting formulation is a rich mixed-integer quadratic linear problem (MIQLP), with prohibitive complexity in case of large networks and multiple CU locations. Therefore we propose a novel two-stage solution process. First, we transform it to an MILP by linearizing its constraints \cite{linearization_gupte}, and then employ an intelligent cutting-planes method based on the seminal Benders' decomposition technique \cite{benders} to find an exact solution. This is crucial since in such network design problems any optimality gap will induce multiplicative (with time) performance and cost losses. 

We evaluate our framework using measurements for the system parameters and operational RAN topologies taken from \cite{fluidran_andrea}. Following an extensive parameter-sensitivity analysis we find that the benefits when departing from single-CU deployments can be as high as 30\%, when using 8 CUs, or 20\%, for 4 CUs in these systems, but these gains diminish as we add more CUs. This reveals a threshold effect, and the need to optimize the placement decisions in order to avoid excessive deployment costs. Finally, we find that a multi-CUs vRAN can support higher level of centralization (more functions at the CUs) even for high routing costs and/or DU traffic loads.




\section{Preliminaries and Model}

\textbf{Background}. 3GPP introduced eight different splits but the main standardized ones are as follows \cite{function_split_bells, adaptive_alba}:
\begin{itemize}[leftmargin=3mm]
    \item Split 1 (S1; PDCP-RLC): RRC, PDCP, and upper layers deployed at CU; RLC, MAC, and PHY at the DU. 
    \item Split 2 (S2; MAC-PHY): MAC and upper layers at CU; PHY and RF layers at the DU. 
    \item Split 3 (S3; PHY-RF): All functions at CU except RF layers.
\end{itemize}
Going from S1 to S3, more functionalities are centralized, which increases the cost savings (economics of scale) and performance gains, e.g., S1 supports CoMP and effective MIMO implementation. However, centralizing more functions increases the volume of data that needs to be transferred to CU (can reach higher than 2.5Gbps in S3), and imposes tighter transfer delay constraints (0.25msec in S3), see \cite{split_3gpp}. 


\begin{figure}
	\centering
	\includegraphics[width=0.45 \textwidth]{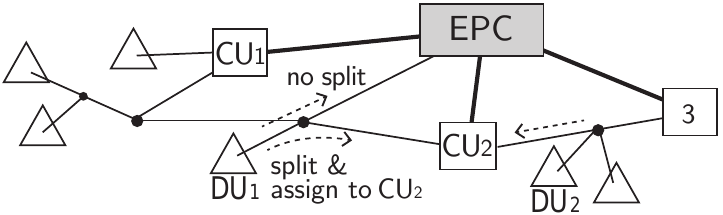}   
	\caption{\small The BS functions can be split between the DU and one of the CUs. DU-$1$ can implement a full-stack BS (no split) and route its traffic directly to EPC; or select a split and send its traffic to CU-$2$ which further routes it to EPC. CUs have high-capacity links to EPC; and some CUs might not be activated (e.g., 3). }
	\label{fig:model}
\end{figure}

The BS implements a function chain $f_0\rightarrow f_1\rightarrow f_2\rightarrow f_3$ \cite{smallcell}, where $f_0$ encapsulates the RF-related operations (e.g., A/D sampling) and is always placed at DUs, while $f_3$ is associated with PDCP and upper layers. Functions $f_1$ (PHY) and $f_2$ (RLC and MAC) can be deployed either in CUs or DUs depending on the RAN configuration. The current standardization efforts suggest a packet-based network with links that are shared by the DUs, instead of the point-to-point expensive CPRI links. Following the latest proposals, we assume that the CUs are directly connected to the network core through high-capacity (e.g., optical) links; see Fig. \ref{fig:model} for an example of our system.


\textbf{Network}. The RAN is modeled with a graph $G = (\mathcal{I,E})$ where the set of nodes $\mathcal I$ includes the subsets: $\mathcal{N}$ of $N$ DUs, $\mathcal{L}$ of $L$ routers, and $\mathcal{M}$ of $M$ possible locations for the CUs ($M << N$); and the core node (EPC) that we index with $0$. We also define $\mathcal{M}_0\!=\!\mathcal{M}\cup \{0\}$. The nodes are connected with the $\mathcal{E}$ links and each link $(i,j)$ has average data transfer capacity of $c_{ij}$ (b/sec). DU-$n$ is connected to each CU-$m$ with paths $\mathcal P_{nm}$, and to EPC with paths $\mathcal P_{n0}$. We define the set $\mathcal{P}_m = \cup_{n=1}^N \mathcal{P}_{nm} $ of paths connecting all DUs to CU-$m$; and the set of all paths $\mathcal{P}= \cup_{m\in\mathcal{M}_0} \mathcal{P}_{m}$. Each path $p_k$ introduces end-to-end delay of $d_{p_k}$ secs. The BS functions are implemented in servers using virtual machines (VMs). We denote with $H_n$ and $H_m$ (cycles/s) the processing capacity of DU-$n$ and CU-$m$, respectively; and define as $\rho_1$ and $\rho_2$ (cycles/Mb/s) the processing load of functions $f_1$ and $f_2$ (per unit of traffic). 

\textbf{Demand \& Cost}. Without loss of generality, we focus on uplink. The users served by each DU-$n$, $n\in\mathcal N$, generate requests following an i.i.d random process $\{ \Phi_n (t) \}^{\infty}_{t=1}$ that is uniformly upper-bounded by $\phi_n$. This creates a data flow $\lambda_n\!=\!\lambda \phi_n$, where $\lambda$ models the bytes per request. Hence, the RAN needs to admit, route and serve $N$ different flows. The function execution is more cost-effective at the CUs \cite{cost_vm}. Therefore, we denote with $a_m$ and $b_m$ the computing and instantiating VM cost (monetary units/cycle) at CU-$m$, and define vectors $\bm{a}=(a_m, m\in\mathcal M)$, $\bm{b}=(b_m, m\in\mathcal M)$. The respective costs for the $N$ DUs are $\bm{\alpha}\!=\!(\alpha_n, n\in\mathcal N)$ and $\bm{\beta}\!=\!(\beta_n, n\in\mathcal N)$. Finally, $\zeta_k$ is the routing cost (monetary units/byte) for each path $p_k\in \mathcal P$, and we define $\bm{\zeta }\!=\!(\zeta_1, \ldots, \zeta_{|\mathcal{P}|})$. Such costs might arise because the links are leased by third-parties, or they can model the (average) operating expenditures of the network. The flows are routed to CUs and then to the EPC through high-capacity links, or directly from the DUs to the EPC in case of no splitting.

\textbf{Problem Definition}. Our objective is to minimize the network operation cost while serving the user traffic, by jointly optimize the following decisions: 
\begin{enumerate}[leftmargin=5mm]
    \item \textit{Deployment}: which of the available CU locations to use?	
    \item \textit{Assignment}: which CU should serve each (split) DU?
    \item \textit{Placement}: where to deploy $f_1,f_2,f_3$ for each BS?    
    \item \textit{Routing}: How to route the data from the DUs to CUs? 
\end{enumerate}
These decisions are inherently coupled and this raises interesting trade-offs. Placing the functions at the DUs reduces the routing cost but increases the computing cost due to inefficient DU servers. On the other hand, centralizing the functions reduces computing costs at the expense of higher routing expenditures, while the deep splits restrict the routing options due to their tight delay bounds. The CU assignment decisions affect directly the routing costs and the number of available paths. Finally, the links are shared, hence the above decisions are coupled across the different DUs. 

\section{Problem Formulation} \label{sec:vran}

\textbf{Function Placement}. We denote with $x_{1n}, x_{2n}\! \in\! \{ 0,1 \}$ the decisions for deploying $f_1$ and $f_2$, respectively, to DU-$n$. Similarly, $y_{1nm}, y_{2nm} \in \{ 0,1 \}$ decide the deployment of these functions to CU-$m$. We define the function placement vectors as $\bm{x}_1 = \left( x_{1n}: n \in \mathcal{N} \right),\bm{x}_2 = \left( x_{2n}: n \in \mathcal{N} \right)$ for all DUs, and $ \bm{y}_{1m} = \left( y_{1n,m}: n \in \mathcal{N} \right), \bm{y}_{2m} = \left( y_{2n,m}: n \in \mathcal{N} \right)$ for CU-$m$. We further define $\bm{y}_1\!=\!\left( \bm{y}_{1m} : m \in \mathcal{M} \right)$ for $f_1$, $\bm{y}_2 = \left( \bm{y}_{2m} : m \in \mathcal{M} \right)$ for $f_2$, and $\bm{x}= \left( \bm{x}_1; \bm{x}_2 \right)$, $\bm{y}= \left( \bm{y}_1; \bm{y}_2 \right)$. 

The function placements are coupled. Namely, $f_1$ cannot be deployed at a CU unless $f_2$ is also placed there; while $f_2$ can be deployed at an DU only after $f_1$ \cite{mec_andrea,fluidran_andrea,smallcell}, i.e.,
\begin{align} 
y_{1nm} &\leq y_{2nm},  \quad \forall n \in \mathcal{N},  m \in \mathcal{M},\label{eq:chain1} \\
x_{2n} &\leq x_{1n}, \quad \,\,\,\,\forall n \in \mathcal{N}. \label{eq:chain2}
\end{align}
Also, duplicate deployments should be prevented, i.e.,
\begin{equation} \label{eq:noduplicate1}
x_{1n} + \sum_{m \in \mathcal{M}} y_{1nm} = 1,\,\,\,\,
x_{2n} + \sum_{m \in \mathcal{M}} y_{2nm} = 1, \  \forall n \in \mathcal{N}.
\end{equation}
%
Finally, the placements need to ensure that the computing capacity at each location is satisfied, hence:
\begin{align}
&\!\lambda_n \big(x_{1n} \rho_1 + x_{2n} \rho_2 + \rho_3(1-\sum_{m} z_{nm}) \big) \leq H_{n},\forall n \in \mathcal{N}, \label{eq:constraint_max_ru} \\
&\sum_{n \in \mathcal{N}} \lambda_n (y_{1nm} \rho_1 + y_{2nm} \rho_2+z_{nm} \rho_3) \leq H_{m}, \forall m \in \mathcal{M}. \label{eq:constraint_max_cu}
\end{align}

\textbf{CU Deployment}. We assign DU-$n$ to CU-$m$ using the binary variable $z_{nm}$, where $z_{nm}=1$ if at least one function of BS-$n$ is deployed at that CU. The assignment matrix is 
$ \bm{z} = (z_{nm} \in \{ 0, 1 \} : \forall n \in \mathcal{N}, \forall m \in \mathcal{M}).$
We model the configuration where $f_1$ and $f_2$ are at the DU-$n$, but $f_3$ at the CU-$m$, by setting $y_{1m}=y_{2m}=1-z_{nm}=0$, hence we do not define explicit variables for $f_3$. Without loss of generality, each DU can be assigned at most to one CU:
\begin{equation} \label{eq:z=1}
\sum_{m \in \mathcal{M}} z_{nm} \leq 1, \forall n \in \mathcal{N}, 
\end{equation}
\begin{equation} \label{eq:z>y}
\text{and}\quad     z_{nm} \geq y_{2nm}, \  \forall n \in \mathcal{N}, \forall m \in \mathcal{M},
\end{equation}
must hold to preserve the function deployment ordering.

\textbf{Data Routing}. Variable $r^{k}_{nm}$ (Mb/s) decides the DU-$n$ traffic routed over path $p_k \in \mathcal{P}_{nm}$ to CU-$m$, and we define $ \bm{r} = \big( r^{k}_{nm}\in\mathbb R_+: \ \forall p_k \in \mathcal{P}_{nm}, n \in \mathcal{N}, m \in \mathcal{M}_0 \big).$
The routing decisions must respect the link capacities 
\begin{equation}
\sum_{n \in \mathcal{N}} \sum_{p_k \in \mathcal{P}_n} r^{k}_{nm}  I_{ij}^{k} \leq c_{ij}, \ \ \forall (i,j) \in \mathcal{E}  \label{eq:link_capacity}
\end{equation}
where it is $I_{ij}^k\!=\!1$ if link $(i,j)$ is included in path $p_k$; and also have to satisfy the flow and delay requirements of the splits: 
\begin{equation} \label{eq:route_assign}
\sum_{p_k \in \mathcal{P}_{nm}} r^{k}_{nm} = z_{nm} S_{n}(x_{1n}, x_{2n}), \
    \forall n \in \mathcal{N}, \forall m \in \mathcal{M},
\end{equation}
where $S_{n}(x_{1n}, x_{2n})$ is the data flow (Mb/s) from DU-$n$ which is determined by the split and user traffic of BS-$n$,
\begin{equation}
	\begin{split} \notag
		S_{n}(x_{1n}, x_{2n}) &= x_{1n}(1.02 \lambda_n+1.5) \\ &  - x_{2n} (0.02 \lambda_n+1.5) + 2500  (1-x_{1n}).
	\end{split}
\end{equation}
\eqref{eq:route_assign} ensures there is no data flow from DU-$n$ to CU-$m$ unless the BS-$m$ $f_3$ is placed at that CU, namely $n$ is assigned to $m$. This captures nicely the interaction between assignment and routing, but creates a quadratic constraint term. 

Note that in case of fully decentralized BSs ($f_3$ at the DUs), the flow needs to be routed directly to the core:
\begin{equation} \label{eq:route_assign2}
	\sum_{p_k \in \mathcal{P}_{n0}} r^{k}_{n0} = \left(1- \sum_{m \in \mathcal{M}} z_{nm}\right) \lambda_n, \
	\forall n \in \mathcal{N}.
\end{equation}

\textbf{Delay}. Variable $r^{k}_{nm}$ has to satisfy the delay requirements of the selected BS split \cite{smallcell}. To enforce this, we first classify the paths into three categories: $\mathcal{P}_{nm}^A \subseteq \mathcal{P}_{nm}$ with delay larger than 30 ms; the paths $\mathcal{P}_{nm}^B \subseteq \mathcal{P}_{nm}$ with delay larger than 2ms; and paths $\mathcal{P}_{nm}^C \subseteq \mathcal{P}_{nm}$ with delay larger than 0.25ms. Clearly, it holds $\mathcal{P}_{nm}^A \subseteq \mathcal{P}_{nm}^B \subseteq \mathcal{P}_{nm}^C$. We can ensure that only the eligible for each split paths are selected by using the constraints:
\begin{align} 
\sum_{p_k \in \mathcal{P}^A_{nm}}\!\! r^{k}_{nm} &\leq T (y_{1nm} +  y_{2nm}), \forall n \in \mathcal{N}, m \in \mathcal{M},  \label{eq:delay1} \\
\sum_{p_k \in \mathcal{P}^B_{nm}}\!\! r^{k}_{nm} &\leq T (1-y_{1nm}+y_{2nm}), \forall n \in \mathcal{N}, m \in \mathcal{M}, \label{eq:delay2} \\
\sum_{ p_k \in \mathcal{P}^C_{nm}}\!\! r^{k}_{nm} &\leq T (2-y_{1nm}+y_{2nm}), \forall n \in \mathcal{N}, m \in \mathcal{M}, \label{eq:delay3}
\end{align}
where $T\!>>\!0$ is used to enforce the logical coupling of the delay requirements of the splits and the eligible paths. 

\textbf{Objective}. We wish to minimize the deployment, computation and routing costs. The computation cost for DU-$n$ is 
\begin{align} \label{eq:cost_ru}
&V_n (\bm{x_1},\bm{x_2},\bm{z}) \!=\!\alpha_n \big(x_{1n} + x_{2n} + (1- \sum_{m \in \mathcal{M}} z_{nm}) \big) \notag \\
&+ (\beta_n \lambda_n) \big(\rho_{1} x_{1n}+\rho_{2}x_{2n}+ \rho_{3}(1- \sum_{m \in \mathcal{M}} z_{nm}) \big),
\end{align}
and the respective cost for CU-$m$ is 
%
\begin{align}
	&V_m (\bm{y}_{1m}, \bm{y}_{2m}, \bm{z}_{m})\!=\! b_m \sum_{n \in \mathcal{N}} \lambda_n \left( \rho_1 y_{1nm} +  \rho_2 y_{2nm} + \rho_3 z_{nm}\right)  \notag \\
	&+ a_m \sum_{n \in \mathcal{N}} (y_{1nm} + y_{2nm} + z_{nm})+\omega_m \sum_{n\in\mathcal{N}} z_{nm}\lambda_n.\label{eq:cost_cu}
\end{align}
where  $\omega_m$ is the cost for using the CU and routing the data from there to the EPC. Note that when it holds $\sum_nz_{nm}\!=\!0$ it means that CU-$m$ is not used. The cost of routing data from the all DUs to CU-$m$ is:
\begin{equation} \label{eq:cost_routing}
U_m(\bm{r}_{m}) = \sum_{n \in \mathcal{N}} \sum_{p_k \in \mathcal{P}_{nm}} \zeta_{k} r_{nm}^{k}.
\end{equation}

Putting the above together, we can introduce the mathematical program $\mathbb{P}_1$ that minimizes the vRAN cost:
\begin{mdframed}
\vspace{-2mm}
\begin{align}
\!\mathbb P_1\!:\!\underset{\bm{r}, \bm{x},\bm{y},\bm{z}}{\text{min}} \  
&\sum_{n\in \mathcal{N}} V_n(\bm{x}_1,\bm{x}_2,\bm{z})
\!+\! \sum_{m\in \mathcal{M}} V_m(\bm{y}_{1m}, \bm{y}_{2m}, \bm{z}_{m})\notag \\
& \!+\!\sum_{m \in \mathcal{M}_0} U_{m}(\bm{r}_{m}) \notag \\
\text{s.t.} \ \ & \eqref{eq:chain1} - \eqref{eq:delay3}\nonumber 
\end{align}
\vspace{-2mm}
\end{mdframed}
$\mathbb P_1$ includes integer, continuous variables, and constraints \eqref{eq:route_assign} involve quadratic terms (variables multiplications). Hence, it is a challenging mixed-integer quadratic program (MIQCP)\footnote{Reduction from the multidimensional-knapsack problem; see also \cite{fluidran_andrea}.}.


\section{Solution Framework}
We follow a two-stage solution approach. First, we reformulate $\mathbb P_1$ using a linearization technique that replaces the intricate constraint \eqref{eq:route_assign}. Then, we decompose the problem and employ an efficient cutting-planes method that expedites the solution and finds, provably, an \emph{exact optimal} point. 

\vspace{-1mm}
\subsection{Linearization of constraints} \label{subsec:linearization}

The product of two integer variables in \eqref{eq:route_assign} can be represented by introducing auxiliary variables for every pair $(n, m)$: $v_{1n,m} = x_{1n} z_{nm}$ and $ v_{2n,m} = x_{2n}  z_{nm}$,
where belong to set:
\begin{equation*}
\begin{split}
\mathcal{V}\!=&\Big\{\! \bm{v_1}\!:\! v_{1nm} \in \{ 0,1 \},\ \bm{v_2}\!:\! v_{2nm} \in \{ 0,1 \}, n\!\in\!\mathcal N, m\!\in\!\mathcal M  \\ 
&\mid v_{1nm} \leq  x_{1n}; v_{1nm} \leq z_{nm}; \ v_{1nm} \geq x_{1n} + z_{nm}\! -\! 1; \\ 
&v_{2nm} \leq  x_{2n}; \ v_{2nm} \leq z_{nm}; v_{2nm} \geq x_{2n} + z_{nm} - 1 \Big\}  ,  
\end{split}
\end{equation*}
using a reformulation similar to \cite{linearization_gupte}, we define the problem:
\begin{mdframed}
	\vspace{-2mm}
\begin{align}
\mathbb{P}_2:\,\,\,\,\,\underset{\substack{\bm{x},\bm{z}, \bm{v} \in \mathcal{V}, \\ \bm{y}, \bm{r} \succeq 0 }}{ \text{min}}   
& \ J_F(\bm{r}, \bm{x}, \bm{y}, \bm{v}) \notag \\
\text{s.t.}\quad & \eqref{eq:chain1}-\eqref{eq:link_capacity}, \eqref{eq:route_assign2}-\eqref{eq:delay3} \nonumber \\
\sum_{p_k \in \mathcal{P}_{nm}}&\!\!\! r^{k}_{nm} = S_{nm}(v_{1nm},v_{2nm}),\ \forall n\!\in\!\mathcal N, m\!\in\!\mathcal M \notag
\end{align}
\vspace{-2mm}
\end{mdframed}
$J_F$ is the objective of $\mathbb{P}_1$, and we set ${S}_{nm}(v_{1nm},\!v_{2nm})\!=\!z_{nm}S_n(x_{1n},\! x_{2n})$. $\mathbb{P}_2$ is equivalent to problem $\mathbb{P}_1$.

\subsection{Decomposition}
We use the Benders' method \cite{benders}, which decomposes $\mathbb P_2$ to a \textit{Master} sub-problem $\mathbb{P}_{2M}$ that optimizes the binary variables for fixed routing; and to a \emph{Slave} program $\mathbb{P}_{2S}$ that optimizes routing for fixed split and assignment decisions:
\begin{align}
\mathbb{P}_{2S}:\,\,\,\underset{\substack{\bm{r}\succeq \bm{0} } } { \text{min}}   
	& \ J_F(\bm{r}, \bm{\bar x}, \bm{\bar y}, \bm{\bar v}) \notag \\
	\text{s.t.}\quad & \eqref{eq:link_capacity}, \eqref{eq:route_assign2}-\eqref{eq:delay3} \nonumber \\
	\sum_{p_k \in \mathcal{P}_{nm}}&\!\!\!\bar r^{k}_{nm} = S_{nm}(\bar v_{1nm},\bar v_{2nm}),\ \ \forall n\in\mathcal N, m\in\mathcal M. \notag
\end{align}
Following the standard practice in Benders techniques, we will be using the dual Slave problem:
\begin{equation}
\mathbb P_{2SD}:\,\,\,	\max_{ \bm{\pi}}\,\, h(\bm{\pi}, \bm{\bar  x}, \bm{\bar y}, \bm{ \bar z}, \bm{\bar v})\quad \text{s.t.} \quad H^\top\bm{\pi}\preceq \bm{\zeta},
\end{equation}
where $\bm{\pi}$ is the vector of dual variables and matrix $H$ collects the necessary coefficients. 

The Master problem $\mathbb{ P}_{2M}$ optimizes the discrete decisions and a proxy continuous variable $\theta\geq 0$:
\begin{align}
	\underset{\substack{\theta, \bm{x}, \bm{z}, \bm{v} \in \mathcal{V}, \bm{y} }}{ \text{min}}   
	& \ J_F(\bm{\bar r}, \bm{x}, \bm{y}, \bm{v}) +\theta \notag \\
	\text{s.t.}\quad & \eqref{eq:chain1}-\eqref{eq:link_capacity} \nonumber \\
	& h(\bm{\pi}^{\xi}, \bm{x}, \bm{y}, \bm{z}, \bm{v}) \leq \theta,\,\,\forall \bm \pi^\xi \in \mathcal{C}_{O} \label{eq:h1} \\
	& h(\bm{\pi}^{\xi}, \bm{x}, \bm{y}, \bm{z}, \bm{v}) \leq 0,\,\,\forall \bm \pi^\xi \in \mathcal{C}_{F} \label{eq:h2}
\end{align}
where (\ref{eq:h1})-(\ref{eq:h2}) are the optimality and feasibility cuts, respectively, which gradually construct the entire constraint set of $\mathbb{P}_2$. The intuition behind this method is that the optimal solution can be found before a full re-construction is built. 

\subsection{Algorithm}
The detailed process is described in Algorithm 1. It finds the optimal solution iteratively. Firstly, it solves the Master problem ($\mathbb{P}_{2M}$) to find the currently optimal integer decision variables ($\bm{x}^{\tau}, \bm{z}^{\tau}, \bm{v}^{\tau}, \bm{y}^{\tau}$)  and surrogate variable ($\theta^{\tau}$) for every iteration $\tau$ (Step 2). These values are used to set the current lower bound $LB^{(\tau)}$ (Step 3). Then, we can solve ($\mathbb{P}_{2SD}$) and get $\bm{\pi}^{\tau}$ by using $\bm{x}^{\tau}, \bm{z}^{\tau}, \bm{v}^{\tau}, \bm{y}^{\tau}$ (Step 4). Then, using the Master problem, we can obtain a new upper bound (Step 5-7). In every iteration, $\mathcal{C}^O$ (if $\mathbb{P}_{2SD}$ is bounded) and $\mathcal{C}_F$ (if $\mathbb{P}_{2SD}$ is unbounded) are enriched with new cuts (Step 8-12); and will be used to solve $\mathbb{P}_{2M}$ in the next iteration. The steps are repeated until upper and lower bound coincide, and reach optimal solution, $UB-LB \longrightarrow 0$ (Step 14). 

\begin{algorithm}[h]  \caption{Decomposition Algorithm}
	\SetAlgoLined
	\DontPrintSemicolon
	\KwInitialize{ \; $\tau=0;\, \mathcal{C}_O^{(0)}=\mathcal{C}_F^{(0)} =\emptyset;\, UB^{(0)}= -LB^{(0)}>>1;\, \epsilon$ } 
	\Repeat{ $UB^{(\tau)}-LB^{(\tau)}\leq \epsilon$ } 
	{
		Solve $\mathbb P_{2M} (\mathcal{C}^{\tau}_O,\mathcal{C}^{\tau}_F)$ to get $ \theta^{\tau}, \bm{x^{\tau},y^{\tau},z^{\tau},v^{\tau}}$ \;
		$LB^{\tau}\!=\!\sum_{n}\! V_n(\bm{x^{\tau}_1},\bm{x^{\tau}_2},\bm{z^{\tau}}) 
		\!+\! \sum_{m}\! V_m(\bm{y}^{\tau}_{1m}, \bm{y}^{\tau}_{2m}, \bm{z}^{\tau}_{m})\! +\! \theta^{\tau}  $  \;
		Solve $ \mathbb P_{2SD}(\bm{x^{\tau},y^{\tau},z^{\tau},v^{\tau}})$ to obtain $\boldsymbol{\pi}^{\tau}$. \;
		\If{$UB^{\tau} < UB^{\tau-1}$}
		{$ UB^{\tau} \!=\!  \sum_{n} V_n(\bm{x}_1^{\tau},\bm{x}_2^{\tau},\bm{z}^{\tau})
			\!+\! \sum_{m} V_m(\bm{y}^{\tau}_{1m}, \bm{y}^{\tau}_{2m}, \bm{z}^{\tau}_{m}) 
			\!+\! h(\boldsymbol{\pi^{\tau}},\bm{x^{\tau},y^{\tau},z^{\tau},v^{\tau}})$;
		}
		\eIf{$h(\boldsymbol{\pi}^{\tau},\bm{x^{\tau},y^{\tau},z^{\tau},v^{\tau}})<\infty$ }
		{$\mathcal{C}^{\tau+1}_O = \mathcal{C}^{\tau}_O \cup \{ \boldsymbol{\pi}^m \} $ \ \% add extreme point;}
		{
			$\mathcal{C}^{\tau+1}_F = \mathcal{C}^{\tau}_F \cup \{ \boldsymbol{\pi}^m \} $. \ \% add extreme ray;
		}
		$ \tau = \tau + 1$. \;
	}
	Optimal configuration, assignment, deployment:
	$\bm{x^*} = \bm{x}^{\tau}; \bm{y^*} = \bm{y}^{\tau}, \bm{z^*} = \bm{z}^{\tau} $  \;
	Optimal routing
	$\bm{r}^*$ $\text{from}$ $\mathbb P_{SD}(\bm{x^{\tau},y^{\tau},z^{\tau},v^{\tau}})$ \;
\end{algorithm}

Algorithm 1 finds, provably, the solution of $\mathbb{P}_2$ (exact if we set $\epsilon=0$.), as the following theorem states.
\begin{theorem}
Algorithm 1 converges in a finite number of iterations to the optimal solution of problem (P2).
\end{theorem}
\begin{proof}	
We prove that the conditions of Theorem 6.3.4 in \cite{mixed-integer} hold for our case.
\begin{enumerate}
	\item $\bm{r}$ is a non-empty and convex set since the domain of every $r_{nm}^{k}$ is in $\mathbb{R}_{+}$. 
	\item $J_F(\bm{r}, \bm{\bar x}, \bm{\bar y}, \bm{\bar v})$ and (\ref{eq:delay1})-(\ref{eq:delay3}) are convex; and (\ref{eq:route_assign})-(\ref{eq:route_assign2}) are linear.
	\item $\bm{r}$ is closed and bounded by (\ref{eq:link_capacity}), and each $r_{nm}^{k}$ has a real value (RU-$n$ traffic) in $\mathbb{R}_{+}$, so the constrains are also continuous for every fixed binary variable.
	\item We have finite number for all parameters and $r$ is bounded by (\ref{eq:link_capacity}), so the solution is also finite. The inequalities and equalities are set at least there is a feasible solution, so, the first-order constraint qualification also holds for the condition and there exist optimal Langrangian multipliers. 
\end{enumerate}
Hence, all conditions of the theorem are satisfied and the convergence does not depend on the initialization.
\end{proof}

\section{Results and Discussion}

In this section we present a battery of numerical tests using real datasets and topologies, aiming to examine the:
\begin{itemize}[leftmargin=4mm]
	\item optimal multi-CUs vRAN configurations in real networks;
	\item benefits of deploying multiple CUs instead of only one;
	\item cost-effectiveness of optimizing the CU locations and DU-CU assignments, compared to non-optimized deployments.
	\item the effect of routing cost and DU traffic on the cost and centralization level for a single and multiple CUs. 
\end{itemize}

\subsection{Network Topology and Evaluation Setup}
We evaluate our model in the actual networks taken from \cite{fluidran_andrea} of Fig. \ref{fig:topology}. RAN N1 consists of a core node, 198 DUs, 15 CU candidate locations, and routers. RAN N2 has a core node, 197 DUs, 15 CU candidate locations, and routers. The network parameters such as delay, location, distance, and link capacities are derived from the actual data or using other measurement studies. We pre-calculate the DU-CU candidate paths by applying the $k$-th shortest path algorithm \cite{kroute}. The distance of DUs to the CUs and the core network ranges from 0.1km to 25km and the respective path delays vary up to 257.61$\mu$s (N1) and 1152.69 $\mu$s (N2). As candidate locations for deploying CUs, we have selected the network nodes with the highest network degree (more central locations)\footnote{These networks do not have CUs, thus we followed this intuitive approach to select candidate locations that we then feed to our optimization framework.}. 


We set the system parameters according to actual testbed measurements and previous works \cite{cost_vm, mec_andrea, fluidran_andrea}.  The default DU load is $\lambda_n=150$ Mbps for each DU\footnote{This value corresponds to 2$\times$2 MIMO, 1 user/TTI, 20 Mhz, 2 TBs of 75376 bits/subframe and IP MTU 1500B.}. For CPU capacity, we use a \textit{reference core} (RC), Intel i7-4770 3.4GHz, and set the maximum computing capacity to 75 RCs for each CU, and to 2 RCs for each DU. We assume that the default cost of CU-$m$ VM instantiation is a half\footnote{This is based on the typical cost of C-RAN BSs that is half of D-RAN BSs, both for macro (\$50K and \$25K) and micro BS(\$20K and \$10K)}. of DU-$n$ ($a_m = \alpha_n/2$) \cite{cost_vm, mec_andrea, fluidran_andrea} but we also explore the impact of different ratios. The CU processing cost is set to $b_m\!=\!0.017 \beta_n$ according to our measurements in \cite{mec_andrea}. Then, the CU deployment cost is assumed at least 30 times higher of the CU processing cost. This value is calculated based on the comparison server setup price (\$20K) and data processing cost (\$653.54) \cite{cost_vm}. Finally, the routing cost per path grows linearly with distance, $\zeta= c_d \times d_{km}$, where $c_d$ is the cost per Km and captures how expensive is each link (can be different for each network).

\begin{figure}[t!]
	\centering
	\begin{subfigure}[b]{0.24\textwidth}
		\includegraphics[width=\textwidth]{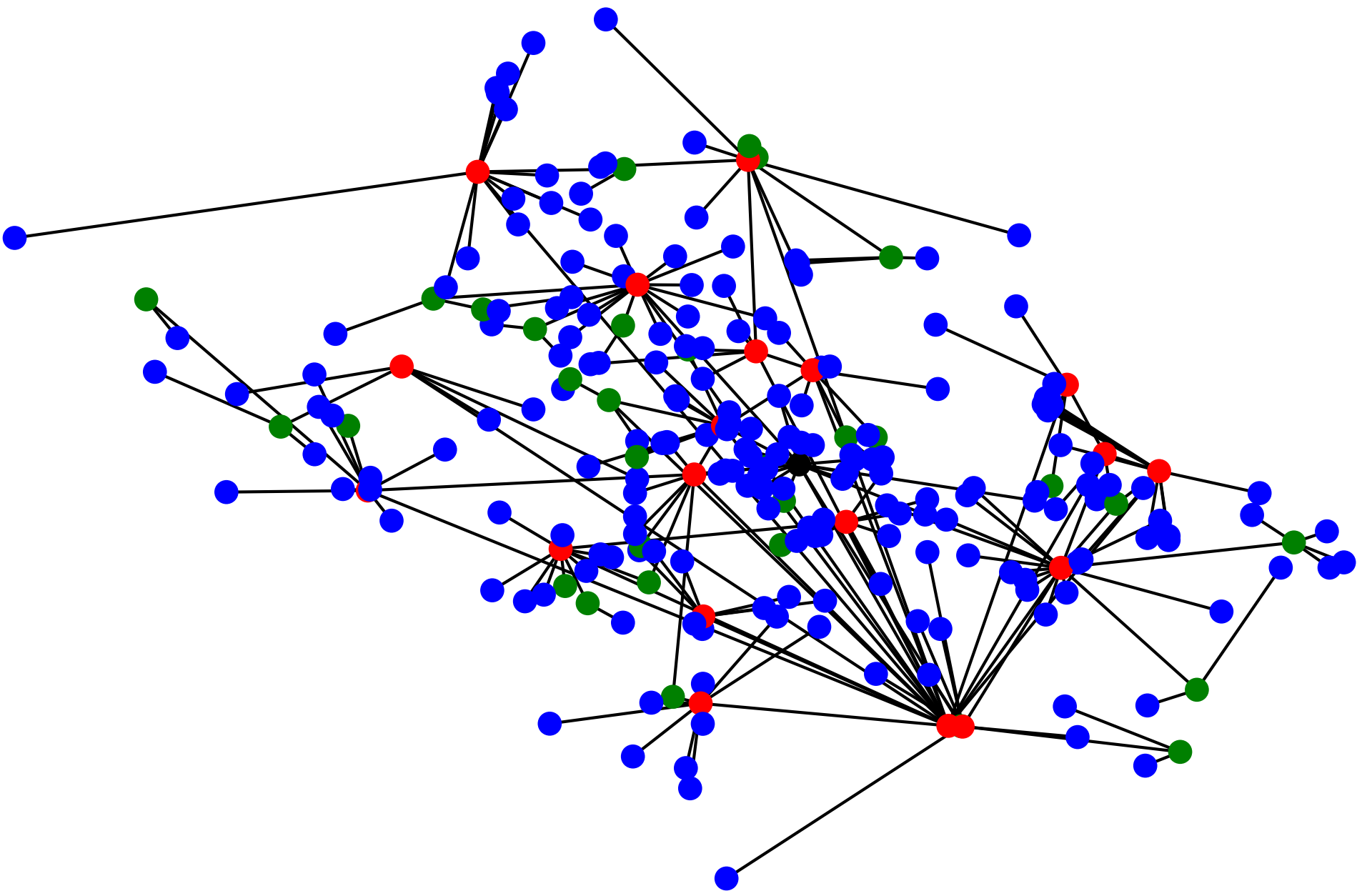}   
		\caption{N1}
		\label{fig:topology_a}
	\end{subfigure}
	\begin{subfigure}[b]{0.241\textwidth}
		\includegraphics[width=\textwidth]{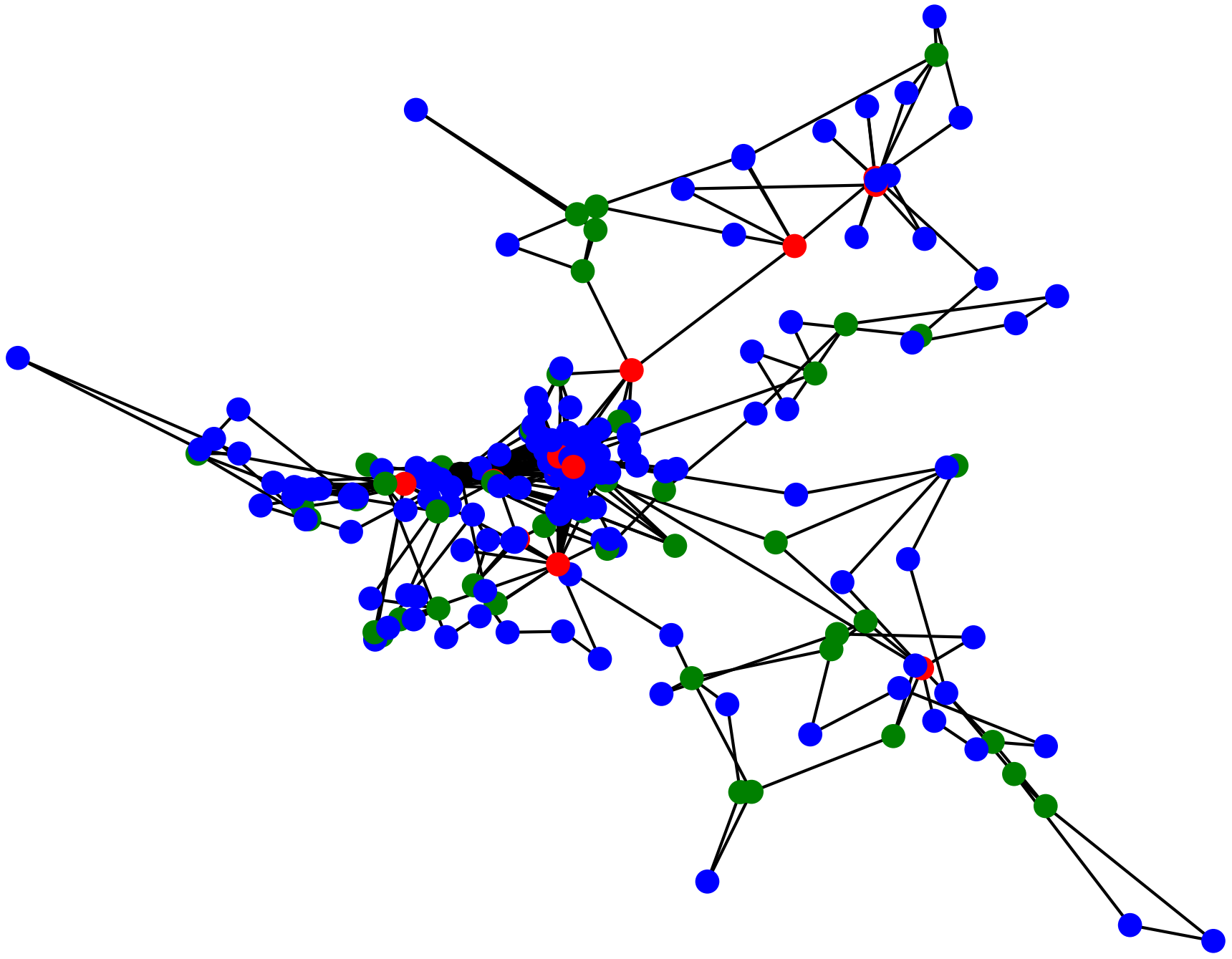}   
		\caption{N2}
		\label{fig:topology_b}
	\end{subfigure}	
	\caption{Two actual operational RANs \cite{fluidran_andrea}. Black, blue, green, and red color dots represent the core network, DUs, routers, and CU's candidates, respectively. The RANs are visualized according to its coordinate location (longitude and latitude).}
	\label{fig:topology}
\end{figure}  

\subsection{Increasing the CU candidate locations}
%
\begin{figure*}[t!] 
	\centering
		\label{fig:num_cu}
	\begin{subfigure}[c]{0.24\textwidth}
		\centering
		\includegraphics[width=\textwidth]{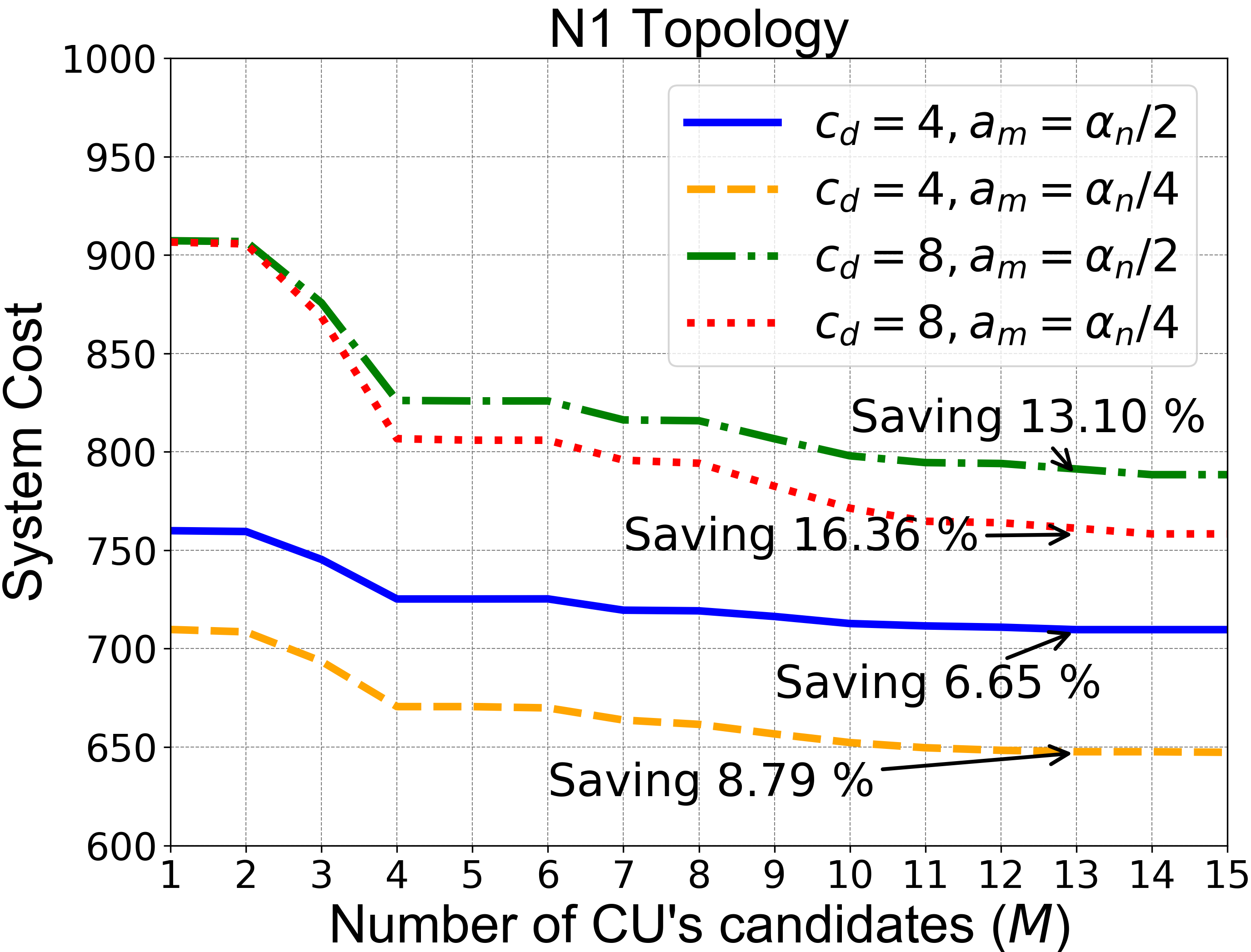}   
		\caption{}
		\label{fig:numofcu_a}
	\end{subfigure}
	\hfill
	\begin{subfigure}[c]{0.24\textwidth}
		\centering
		\includegraphics[width=\textwidth]{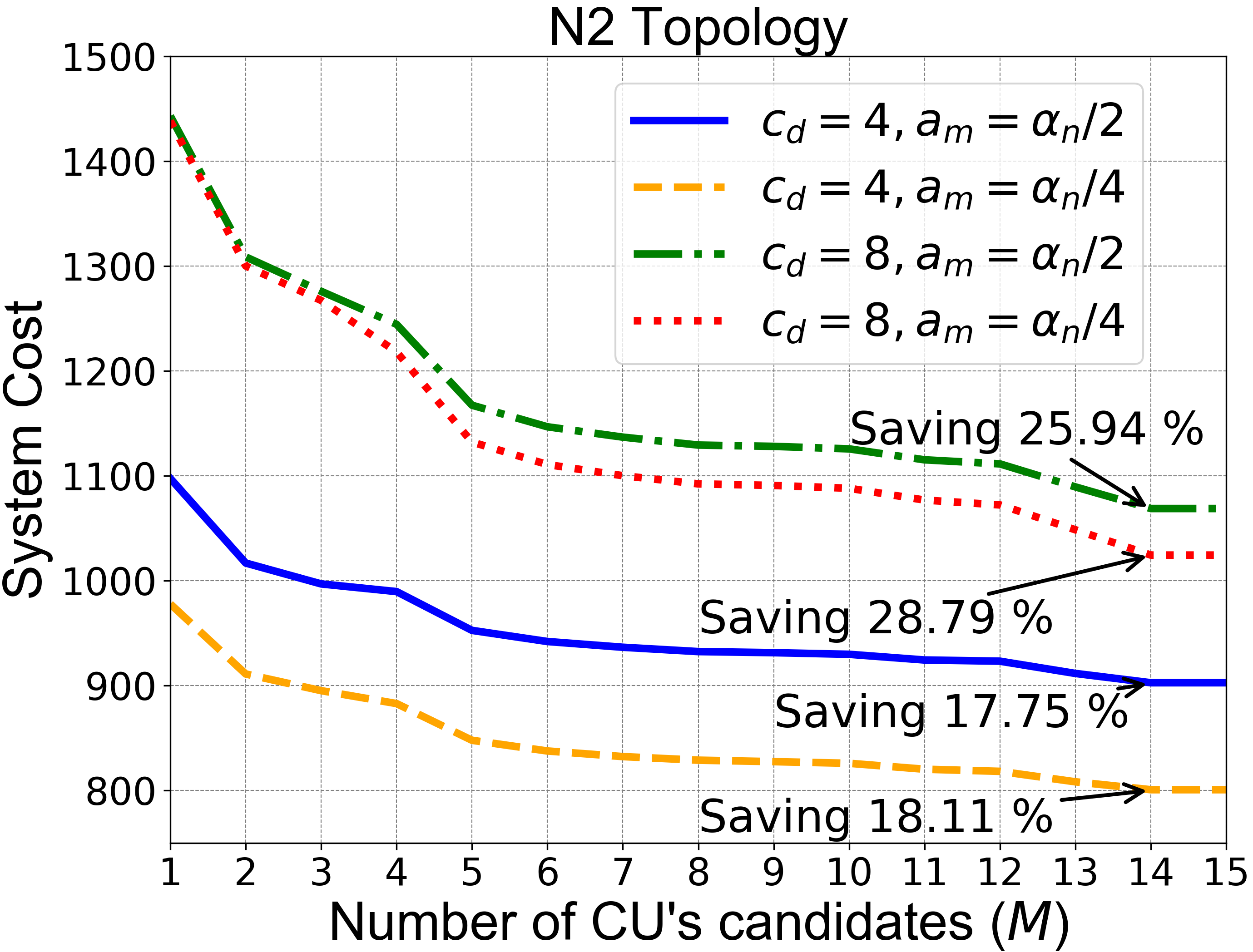}   
		\caption{}
		\label{fig:numofcu_b}
	\end{subfigure}
	\hfill
	\begin{subfigure}[c]{0.24\textwidth}
		\includegraphics[width=\textwidth]{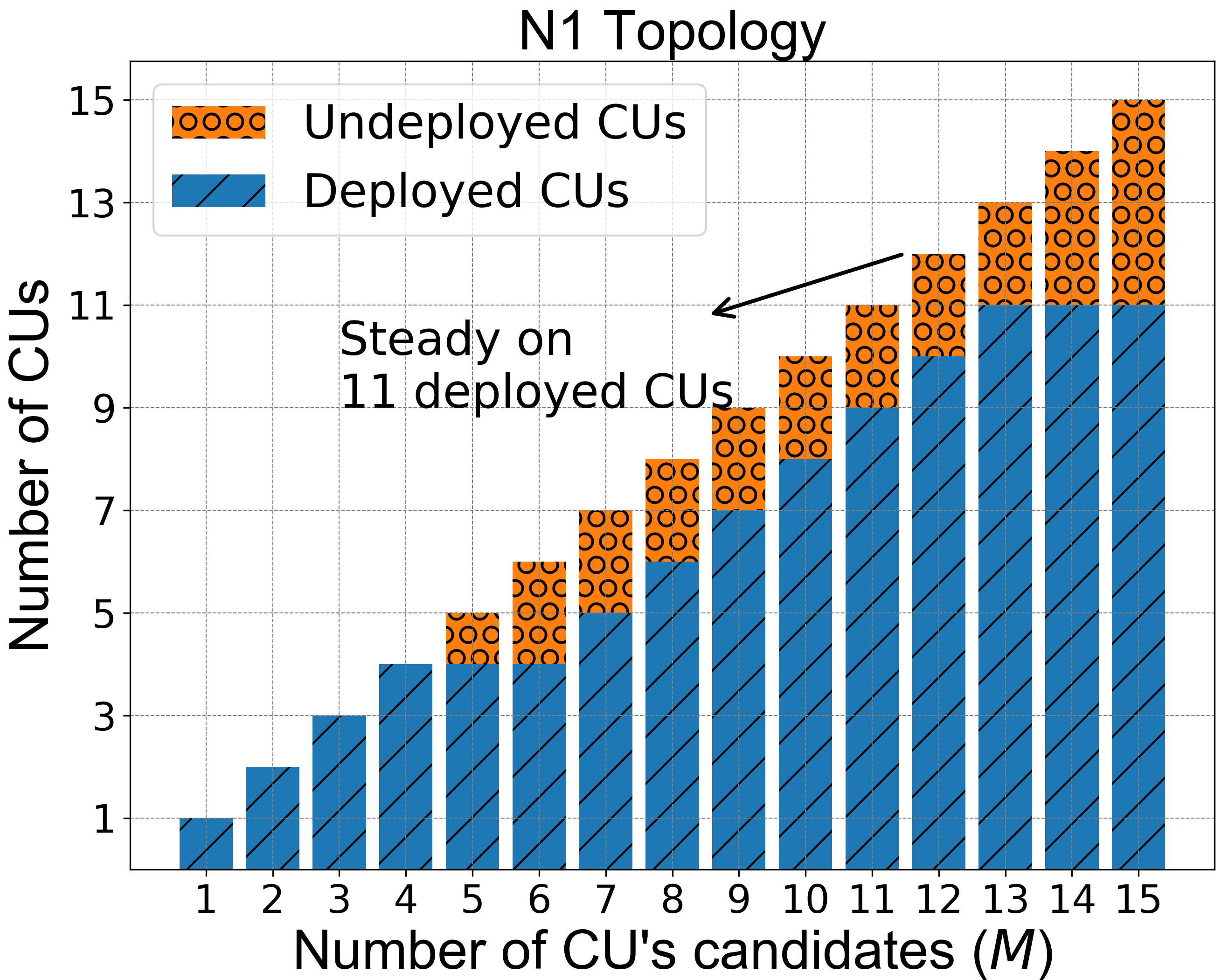}   
		\caption{}
		\label{fig:numofcu_c}
	\end{subfigure}
	\hfill
	\begin{subfigure}[c]{0.24\textwidth}
		\includegraphics[width=\textwidth]{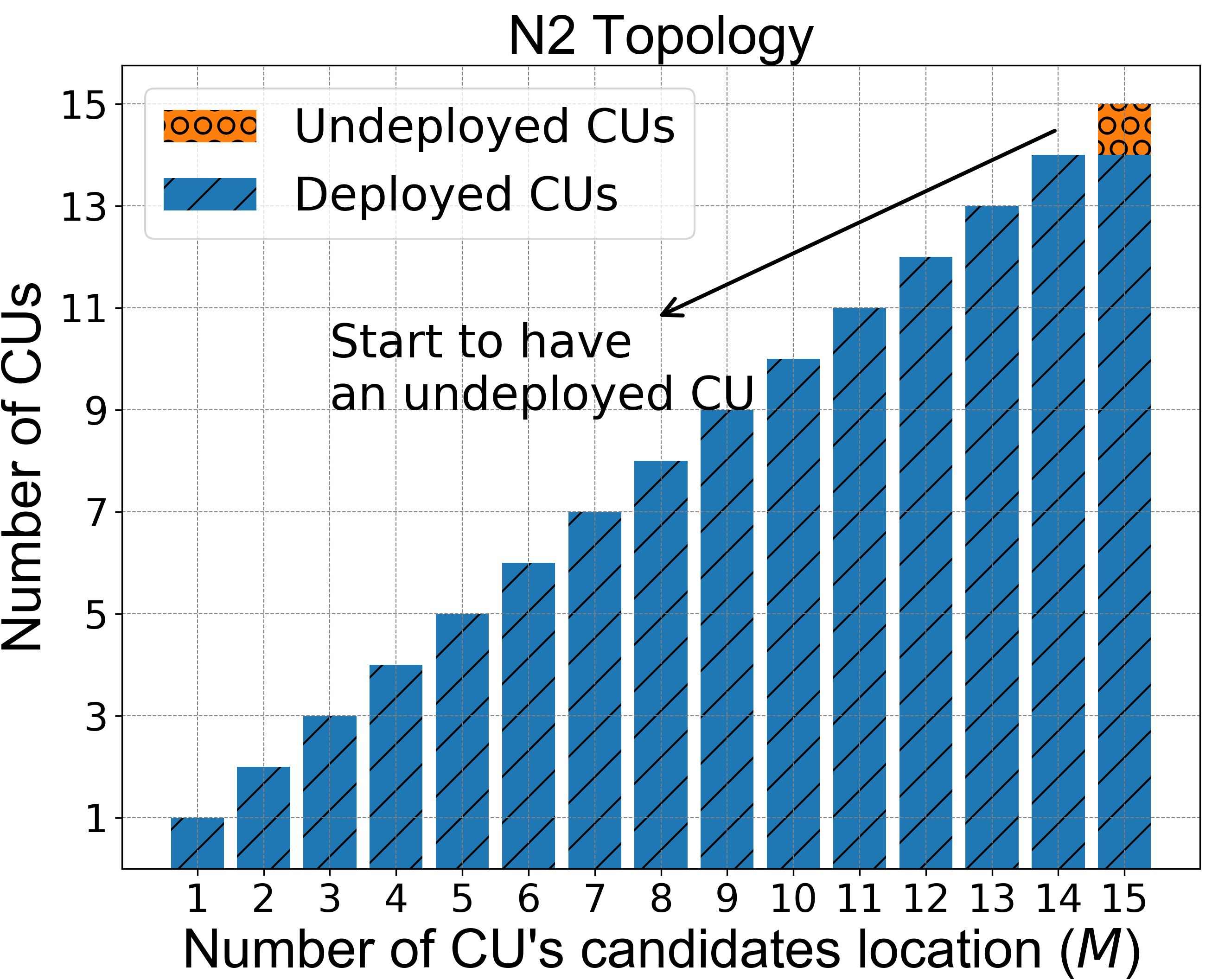}   
		\caption{}
		\label{fig:numofcu_d}
	\end{subfigure}
	\caption{Impact of increasing locations to: (a-b) system cost; (b-c) the number of actually deployed CUs.}
\end{figure*} 
\begin{figure*}[t!] 
	\centering
	\label{fig:transport}
	\begin{subfigure}[c]{0.24\textwidth}
		\centering
		\includegraphics[width=\textwidth]{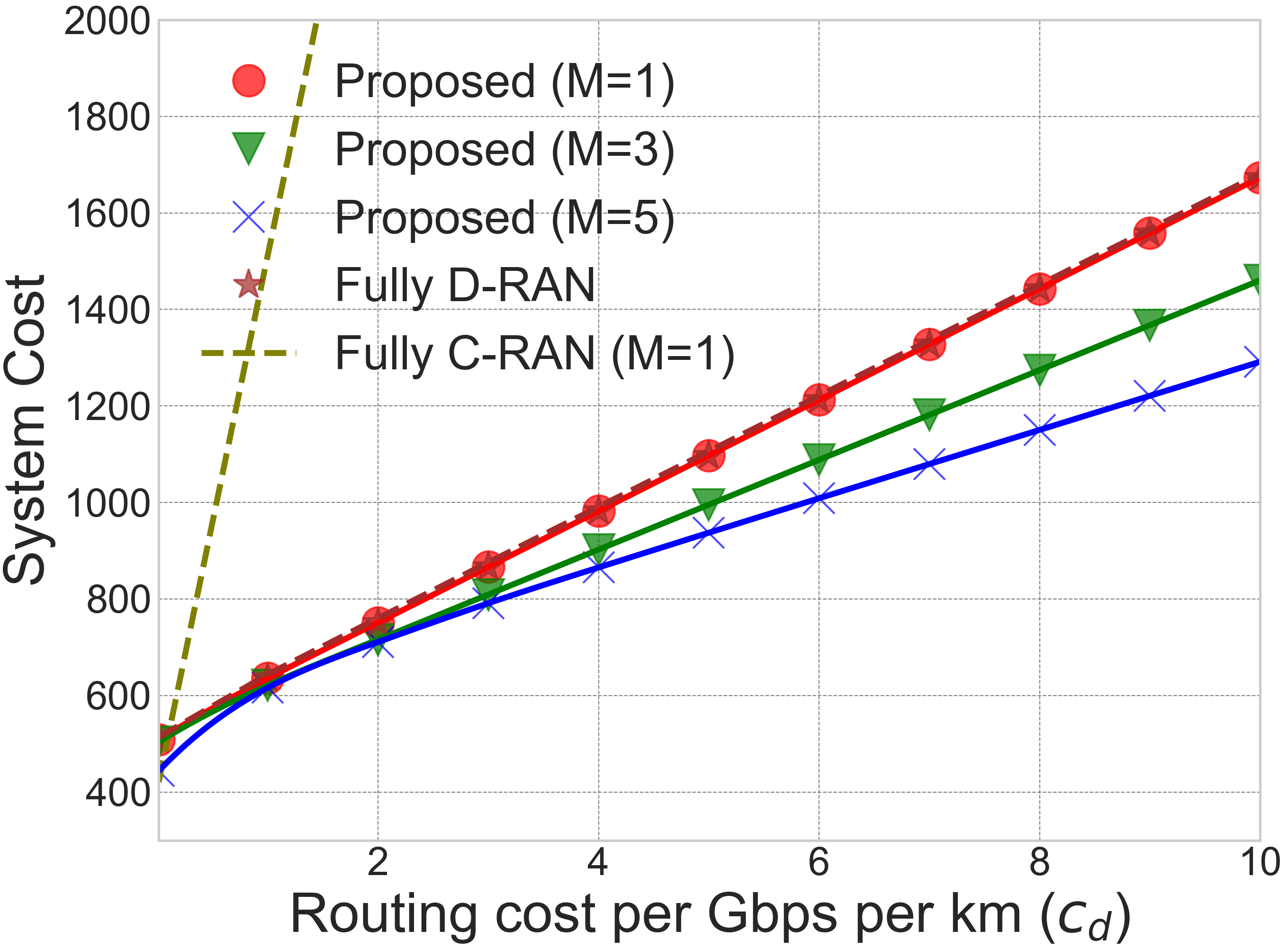}
		\caption{}
		\label{fig:transport1}
	\end{subfigure}
	\hfill
	\begin{subfigure}[c]{0.24\textwidth}
		\centering
		\includegraphics[width= \textwidth]{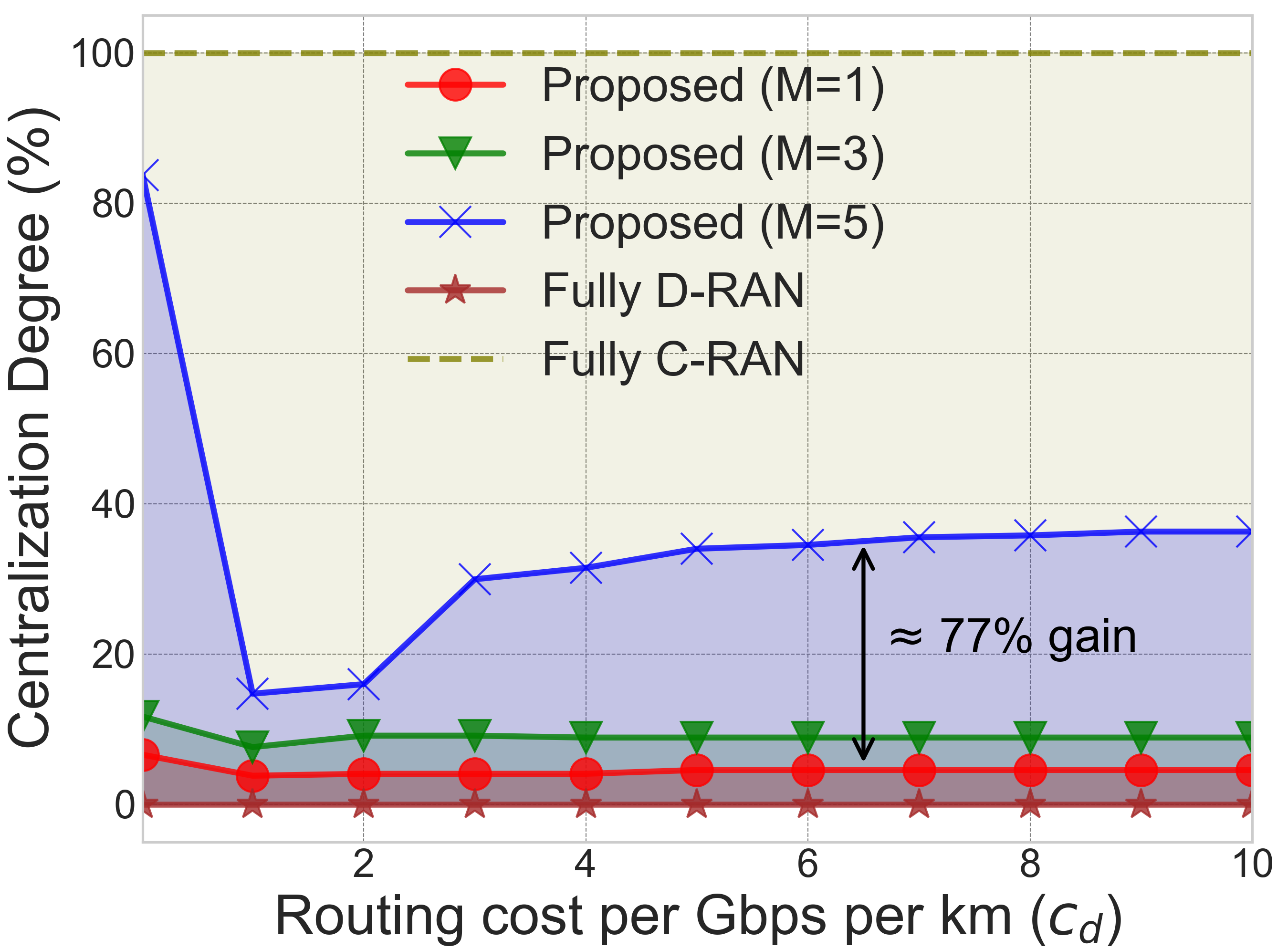}
		\caption{}
		\label{fig:transport2}
	\end{subfigure}
	\hfill
	\begin{subfigure}{0.24\textwidth}
		\centering
		\includegraphics[width=\textwidth]{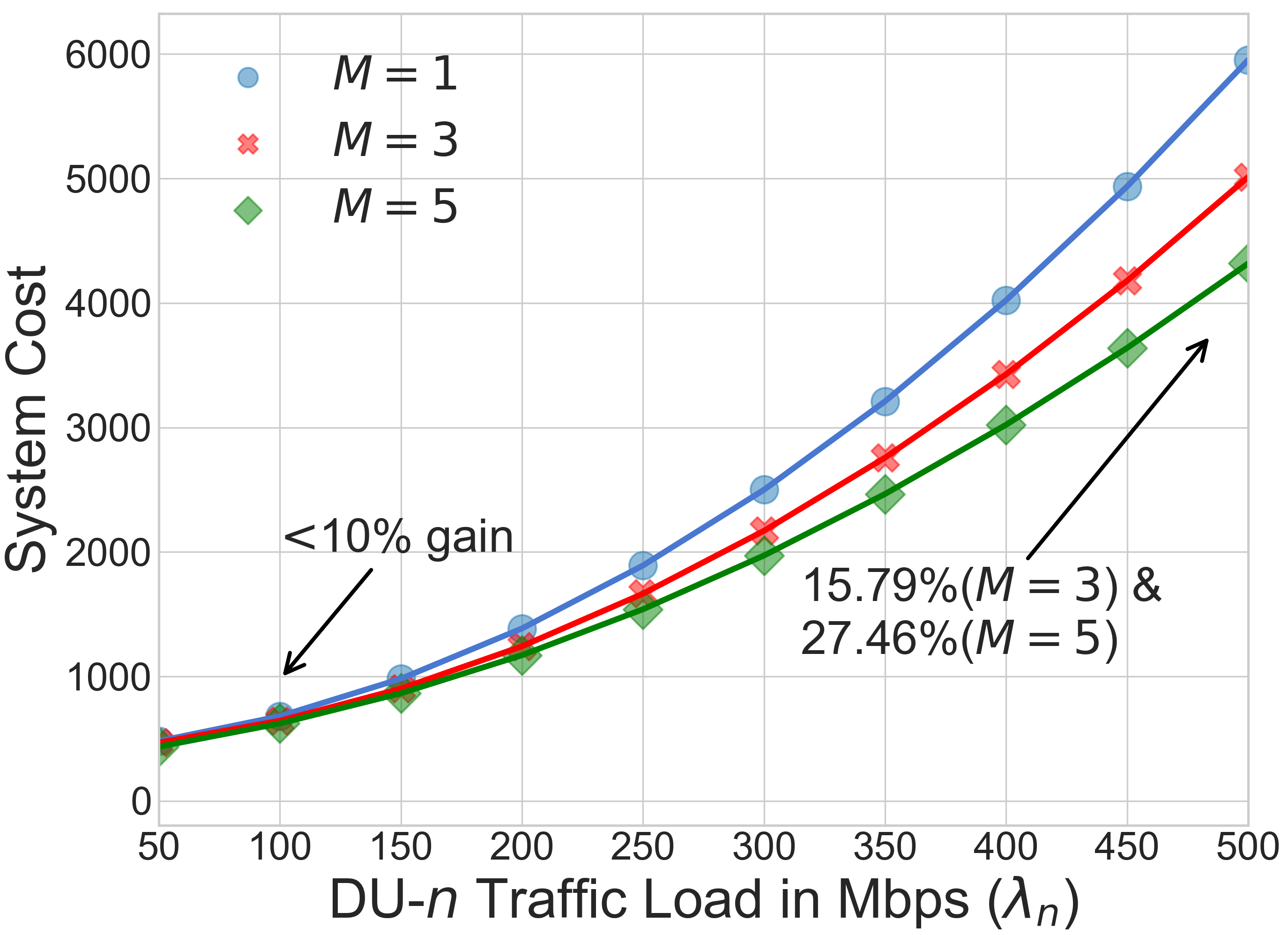}   
		\caption{}
		\label{fig:sigma_b} 
	\end{subfigure}
	\hfill
	\begin{subfigure}[c]{0.24\textwidth}
		\includegraphics[width=\textwidth]{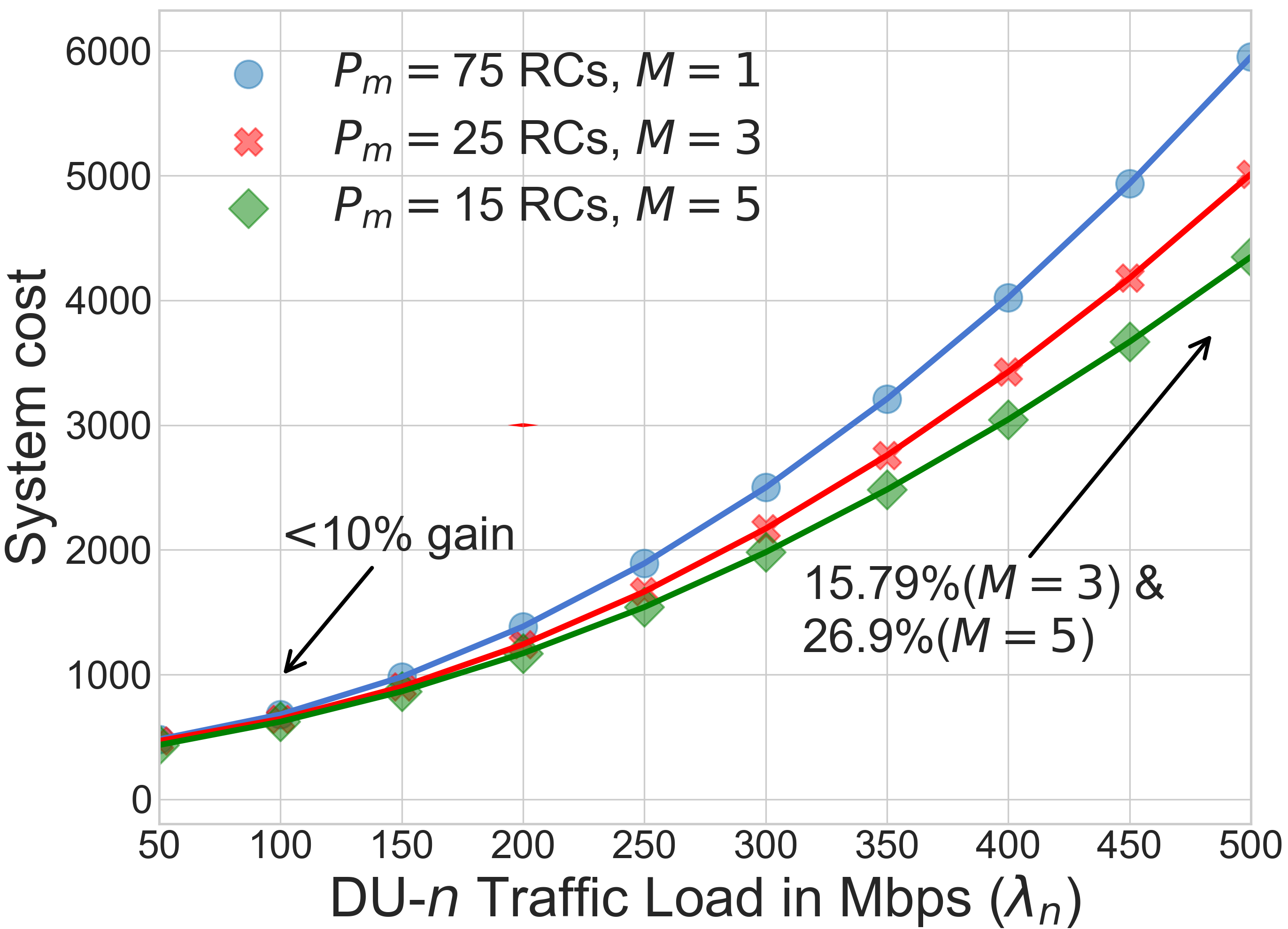}   
		\caption{}
		\label{fig:sigma_a}
	\end{subfigure}
	\caption{Impact of the routing cost to (a) system cost; (b) centralization degree. (c)-(d) Impact DU traffic load to system cost.}
\end{figure*} 

Our first experiment increases gradually the available CU locations and studies how this impacts the system cost (in monetary units) and the RAN design. The deployment cost per CU is higher as we add more locations, i.e., $\omega_m \leq \omega_{m+1}$, which is intuitive as, naturally, the cheapest locations are selected first. Fig. \ref{fig:numofcu_a} shows that there are $6.65\%, 8.79\%, 13.10\%$, and $16.36\%$ cost savings on average by adding more candidate location ($M\!=\!13$) compared to just a single location in N1, for various scenarios of routing cost and CU/DU computing efficiency gains. Similar findings hold for N2, Fig. \ref{fig:numofcu_b}. The trend of cost reduction continues as the candidate CU locations increase, albeit the additional gains are fast decreasing and we reach a state (after $M\!=\!13$ in this network) where adding more CU locations does not bring further cost savings.  

Figures \ref{fig:numofcu_c}-\ref{fig:numofcu_d} show the distribution of deployed CUs and non-deployed ones, i.e., available locations that were not used by our solution. In N1, the deployment rate (portion of available CU locations that are actually deployed) is $100\%$ when the candidate locations are up to $M\!=\!4$. Beyond that, the deployment decreases and reaches a constant value with 11 deployed CUs when the available locations are $M\!=\!13$. However, in N2, we have non-deployed CUs after $M\!=\!14$ which shows that, despite following a similar trend, the deployment decision is profoundly affected by the specific candidate locations and network structure. These results are representative as we have used real networks. 

\textbf{Findings:} \textit{1) As the CU available locations increase, our framework saves significant costs (16.36$\%$ in N1 and 28.79$\%$ in N2), before the gains diminish. 2) The higher routing costs and lower compute costs of CUs lead to substantial benefits. 3) Network characteristics affect the number of used CUs.}



\subsection{Impact of Routing Cost and Traffic}


%
%
We study the effect of routing cost on the system cost and the centralization that our solution achieves i.e., number of functions deployed at the CUs. Aside from comparing single and multiple CU candidate locations, we also benchmark against C-RAN and D-RAN solutions (two extreme cases). In fully C-RAN, all functions except $f_0$ are placed at the CU, while in D-RAN they are all deployed at the DU. In this case, the deployment cost of CU is set to zero ($\omega_m =0$) for D-RAN (we do not need to deploy any CU), and the traffic load of each DU is directly routed to the core network. In this experiment, the routing cost (/Km) increases from $c_d=0.01$ (very low) to $c_d=10$ (very high). Note that networks N1 and N2 cannot support fully C-RAN solutions, hence the results for this configuration are hypothetical and presented for reference. 

Fig. \ref{fig:transport1} compares the system cost of our framework to fully C-RAN and D-RAN. The C-RAN cost increases significantly with the routing cost $c_d$, and eventually exceeds the cost of D-RAN (for $c_d \approx 0.01$ per Gbps). For a single CU location and the specific network parameters\footnote{This result is affected by, e.g., the cost of deploying a CU, the distance of this location compared to the core, etc. In general, even a single CU location can significantly reduce costs, see \cite{fluidran_andrea} for examples.}, the optimized configuration is just slightly better than D-RAN ($0-1\%$) because the CU location is near to the core and far from the DUs, thus most functions cannot be deployed in CU, Fig. \ref{fig:transport2}. By considering more candidate locations we can obtain approximately $13.10\%$ ($M=3$) and $23.15\%$ ($M=5$) cost savings at $c_d = 10$. Fig. \ref{fig:transport2} shows the impact of the routing cost to function centralization\footnote{Percentage of all BS functions that are deployed at CUs; it is 100\% for fully C-RAN and 0$\%$ for D-RAN.}. Overall, with $M=5$ locations we can maintain the centralization compared to the lower candidates. The result shows around $6.2-72.09\%$ and $10.9-77.16\%$ centralization gains at routing cost $0.01-10$ per Gbps/Km distance. 
%

Next, we evaluate the impact of the DU-$n$ traffic load on the system cost. We consider two scenarios. Firstly, the computational capacity of each CU is set with the same value, 75 RCs. Secondly, the aggregate of all CUs capacity
is set to be $P_{tot}$ = 75 RCs. Consequently, every CU has different capacity, i.e., $P_m = P_{tot}/M$, in each scenario as $M$ changes. Our goal here is to understand which design option is preferable: to have a single CU location with very high capacity server, or multiple smaller CUs with lower capacity? We use the same deployment and VM computational cost for both scenarios. Fig. \ref{fig:sigma_b} clearly shows that the multiple CUs scenario performs better. Indeed, we can gain around $15.79\%$ ($M=3$) and $27.46\%$ ($M=5$) cost reduction at 500Mbps, when we consider three ($M=3$) and five CU locations ($M=5$) compared to a single CU (total computing capacity is the same in both cases). The gap is getting higher as the traffic load increases. In addition, Fig. \ref{fig:sigma_a} shows that having multiple candidate locations yields extra saving of $15.79\%$ ($M=3$) and $26.9\%$ ($M=5$) compared to a single candidate, even when the capacity of servers is smaller.

\textbf{Findings:} \textit{1) The cost-savings of having multiple CUs is getting higher in line with the increase of the routing cost and the traffic load. 2) More CU locations can maintain a higher centralization degree (77$\%$ gain). 3) Having multiple CUs with lower capacity can achieve extra cost-savings (26.9$\%$) than a single CU with the same aggregate capacity.}


%
\begin{figure}[t]
	\centering
	\includegraphics[width=0.375\textwidth]{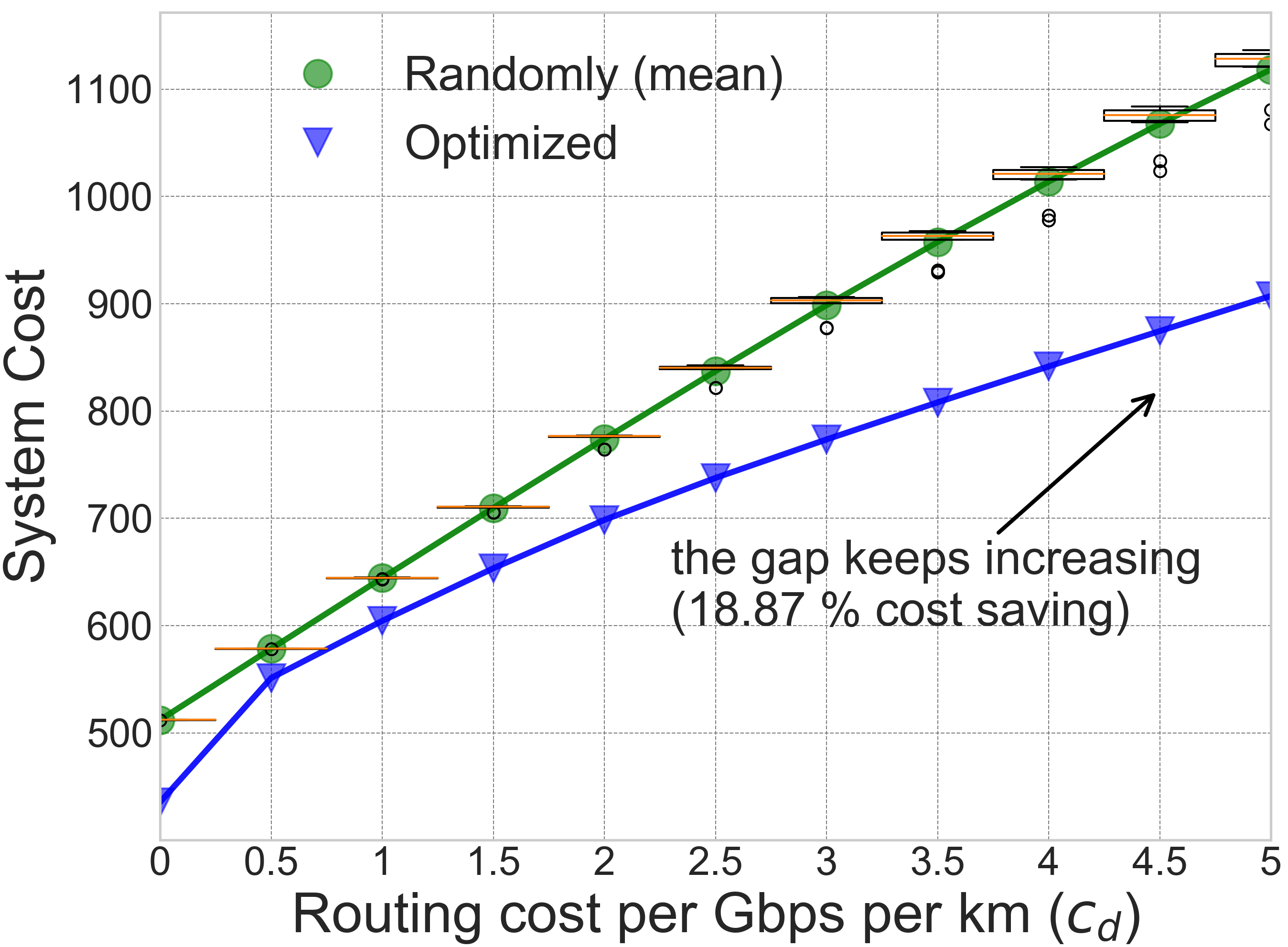}   
	\caption{The comparison of randomly and optimized (proposed) deployment as the routing cost increases.}
	\label{fig:opt} 
\end{figure}

\subsection{Random vs Optimized Deployment of CUs}
Finally, we compare the case where the CUs are deployed randomly instead of optimizing these decisions (as our algorithm does). We will show that it is crucial not only to use multiple CUs but to carefully optimize their placement. The first approach deploys the CUs in random locations (among those being available), while the second uses our framework and optimizes their position. In detail, we set the number of available CU positions to $M\!=\!15$ and force to choose only 3 deployed CUs, by selecting randomly (10 iterations) or by optimizing the deployment. Fig. \ref{fig:opt} clearly shows that our framework outperforms the randomly deployed scenario by 18.87$\%$. These cost savings increase with the routing cost. Clearly, when it is more expensive to route data, such non-optimized random decisions will have higher negative impact. 

\textbf{Finding:} \textit{The gains of multi-CUs vRAN cannot be maximized unless their number and location is optimized (18.87$\%$ gain at $c_d=5$) for the specific network and load.}




\section{Conclusion} \label{sec:conclusion}
There is currently a flurry of standardization and other industry activities, aiming to make vRAN the de facto solution for next generation access networks. To this end, our work fills an important gap as it tackles the vRAN design problem in its most general form. Using a standards-compatible system model, we develop a rigorous optimization approach that selects jointly the number and location of (multiple) CUs; assigns to them the DUs; and finds the optimal split levels and routing paths for each flow. Our framework is general and can be readily tailored to different scenarios, e.g., when the operator needs to enforce some level of centralization or when the cost functions are convex non-linear.


\bibliographystyle{IEEEtran}
\bibliography{IEEEabrv,ref}



\end{document}